\documentclass[12pt,a4paper,reqno]{amsart}
\usepackage{amssymb}
\usepackage{amscd}
\usepackage{amsthm}
\usepackage[pdftex,pdfpagelabels]{hyperref}
\usepackage{enumerate}
\usepackage{graphicx}
\usepackage{mathtools}
\usepackage{siunitx}
\usepackage{tikz-cd}
\usepackage{bm}
\usepackage{mathtools}%                  http://www.ctan.org/pkg/mathtools
\usepackage[tableposition=top]{caption}% http://www.ctan.org/pkg/caption
\usepackage{booktabs,dcolumn}%           http://www.ctan.org/pkg/dcolumn + http://www.ctan.org/pkg/booktabs
\usepackage{subfigure}
\usepackage{caption}
\usepackage{arydshln}

\newcommand\R{\mathbb{R}}

\newcommand\C{\mathbb{C}}

\newcommand\dist{{\operatorname{dist}}}

\newtheorem{theorem}{Theorem}[section]
\newtheorem{lemma}[theorem]{Lemma}

\author{Larry Guth}
\address{Department of Mathematics\\
Massachusetts Institute of Technology\\
Cambridge, MA 02139, USA}
\email{lguth@math.mit.edu}

\author{Dominique Maldague}
\address{Department of Mathematics\\
Massachusetts Institute of Technology\\
Cambridge, MA 02139, USA}
\email{dmal@mit.edu}

\author{John Urschel}
\address{Department of Mathematics\\
Massachusetts Institute of Technology\\
Cambridge, MA 02139, USA}
\email{urschel@mit.edu}


\title{Estimating the matrix $p \rightarrow q$  norm}

\begin{document}

\maketitle

\begin{abstract}
The matrix $p \rightarrow q$ norm is a fundamental quantity appearing in a variety of areas of mathematics. This quantity is known to be efficiently computable in only a few special cases. The best known algorithms for approximately computing this quantity with theoretical guarantees essentially consist of computing the $p\to q$ norm for $p,q$ where this quantity can be computed exactly or up to a constant, and applying interpolation. We analyze the matrix $2 \to q$ norm problem and provide an improved approximation algorithm via a simple argument involving the rows of a given matrix. For example, we improve the best-known $2\to 4$ norm approximation from $m^{1/8}$ to $m^{1/12}$. This insight for the $2\to q$ norm improves the best known $p \to q$ approximation algorithm for the region $p \le 2 \le q$, and leads to an overall improvement in the best-known approximation for $p \to q$ norms from $m^{25/128}$ to $m^{3 - 2 \sqrt{2}}$.

%We use elementary inequalities from analysis to improve the polynomial time approximations for the matrix $2\to q$ norm $\|A\|_{2\to q}$ in the hypercontractive range $q\ge 2$. Our algorithm gives a multiplicative error of $m^{\frac{q-2}{2q(q-1)}}$, compared to the previous best bound of $m^{\frac{q-2}{q^2}}$ obtained via interpolation between $\|A\|_{2\to2}$ and $\|A\|_{2\to \infty}$. %the previous best which Via interpolation, this leads to improvements for best known polynomial time estimates of $\|A\|_{p\to q}$ for some range of $p$ and $q$. 

%Feel free to edit/make more precise
\end{abstract}

\section{Introduction} 
Given a matrix $A \in \mathbb{C}^{m \times n}$, the $p \rightarrow q$ norm of $A$, defined by
$$\|A\|_{p\rightarrow q} := \max_{x \in \mathbb{C}^n} \frac{\|Ax\|_q}{\|x\|_p}, \qquad p,q \in [1,\infty),$$
where $\|\cdot\|_p$ is the $\ell_p$-norm of a vector, is a measure of how large in magnitude a vector of a given $p$-norm can become in $q$-norm from the application of $A$. The matrix $p \rightarrow q$ norm is a fundamental quantity, appearing in a variety of areas of mathematics and computer science. The spectral norm ($p = q =2$) is perhaps the most used and well-known instance of a $p \rightarrow q$ norm, and is equal to the largest singular value of a matrix. More generally, matrix $p \rightarrow p$ norms (simply called $p$-norms) appear in relative error estimates for linear systems $Ax =b$, as the $p\rightarrow p$ condition number of a matrix governs the worst-case relative error for $x$ in $p$-norm due to small perturbations to $A$ and $b$; see both \cite[Chapter 2.6]{golub2013matrix} and \cite[Chapter 14]{higham2002accuracy} for details. The special case of $p = \infty$ and $q =1$ is the well-studied Grothendieck problem \cite{grothendieck1996resume,pisier2012grothendieck}: maximize $\langle y, Ax \rangle$ subject to $\|x\|_\infty, \|y \|_\infty =1$, of which the max-cut problem in graph theory is a special case \cite{goemans1995improved}. The case of $p =2$ and $q = 4$ is also of special significance. For instance, distinguishing between entangled and separable states in quantum information theory is known to be related to the computation of the $2 \rightarrow 4$ norm \cite{harrow2019limitations}. The $2 \rightarrow q$ norm, for any even $q\ge 4$, is directly related to small-set expansion in graph theory and the unique games conjecture in computational complexity \cite{barak2012hypercontractivity}. More generally, norms in the hyper-contractive region of $p < q$ are closely related to mixing of Markov chains and a variety of problems in theoretical computer science, see \cite{gine1997lectures,biswal2011hypercontractivity} for details and further examples. The most general setting of arbitrary $p$ and $q$ is also broadly applicable to robust optimization, see \cite[Section 1.1.2]{steinberg2005computation} and the examples and references contained therein.

One of our initial questions about $\|A\|_{p\to q}$ was whether techniques proving $L^p$ to $L^q$ boundedness from analysis could be applied to improve approximation algorithms for $\|A\|_{p\to q}$. Consider, for example, the Fourier transform $\widehat{f}(\xi)=\int_{\R^n}e^{2\pi i x\cdot\xi}f(x)dx$ of a function $f:\R^n\to\C$. By the classical Hausdorff-Young inequality, $C_{p,q,n}:=\sup_{\|f\|_{L^p}=1}\|\widehat{f}\|_{L^q}$
is finite if and only if $1\le p\le 2$ and $\frac{1}{p}+\frac{1}{q}=1$. The optimal constant $C_{p,q,n}$ is achieved by Gaussian functions, which was proved by Babenko \cite{MR0138939} for even $q$ and then Beckner \cite{MR0385456} for all $2\le q$. %Lieb identified $A_{q,n}$ for even exponents $q$ and Beckner characterized $A_{q,n}$ for all $2\le q$. 
In \cite{sharpenedHY}, Christ sharpened the Hausdorff-Young inequality by proving that  
\[ \|\widehat{f}\|_{L^q} \le [C_{p,q,n}-\dist_p(f,\mathcal{G})]\|f\|_{L^p} \]
for an appropriate distance function $\dist_p(f,\mathcal{G})$ of $f$ to the set of Gaussians. One basic idea we use in our approximation of $\|A\|_{p\to q}$ involves the following fundamental maximization principle that Christ used in \cite{sharpenedHY}. Suppose that $\|f\|_{L^p}=1$ and that for a small parameter $\delta>0$, $C_{p,q,n}(1-\delta)\le \|\widehat{f}\|_{L^q}$. If for another small parameter $\eta>0$, $f=g+h$ is a disjointly supported decomposition of $f$ with $\|h\|_{L^p}\le \eta$, then 
\[  C_{p,q,n}(1-\delta)\le \|\widehat{f}\|_{L^q}\le \|\widehat{g}\|_{L^q}+C_{p,q,n}\|h\|_{L^p}\le \|\widehat{g}\|_{L^q}+C_{p,q,n}\eta,  \]
so $C_{p,q,n}(1-\delta-\eta)\|g\|_{L^p}\le \|\widehat{g}\|_{L^q}$. The basic principle is that if $f$ has close to maximal ratio ${\|\widehat{f}\|_{L^q}}/{\|f\|_{L^p}}$ and $g$ is a significant portion of the $L^p$ mass of $f$, then $g$ also has close to maximal ratio ${\|\widehat{g}\|_{L^q}}/{\|g\|_{L^p}}$. We use this idea in the proof of the key lemma (Lemma \ref{lm:2toq}) below. 

Computing the matrix $p \rightarrow q$ norm consists of maximizing a convex function over a convex domain, and, unfortunately, is difficult in general. The matrix $p \rightarrow q$ norm is only efficiently computable (up to arbitrary error) in three special cases: $p = q = 2$, $p = 1$, and $ q= \infty$. As previously mentioned, the spectral norm $\|A\|_{2 \rightarrow 2}$ is equal to the largest singular value $\sigma_{\max}(A)$ of $A$ and is achieved by any right singular vector $v$ corresponding to $\sigma_{\max}(A)$. Let $r_j$, $j = 1,...,m$, denote the conjugate transpose of the rows of $A$. When $q =\infty$, $\|A\|_{p \rightarrow \infty} = \max_{j \in \{1,...,m\}} \|r_j\|_{p^*}$, where $1/p + 1/p^* = 1$, and this value is achieved by the vector $v=(|r_j(1)|^{p^*-2}r_j(1),\ldots,|r_j(n)|^{p^*-2}r_j(n))$, where $r_j=(r_j(1),\ldots,r_j(n))$ is a row vector with maximal ${p^*}$ norm. In fact, the case of $p = 1$ is mathematically equivalent to $q = \infty$, as
$$ \|A\|_{p \rightarrow q} = \max_{\substack{\|x\|_p = 1, \\ \|y \|_{q^*} = 1}} \langle y, A x \rangle = \|A^*\|_{q^* \rightarrow p^*}, \qquad \text{where} \quad \frac{1}{p} + \frac{1}{p^*} = \frac{1}{q} + \frac{1}{q^*} = 1 .$$
More generally, by duality, it suffices to only consider $\|A\|_{p \rightarrow q}$ for $1/p + 1/q \le 1$. For this reason, our results focus on the regimes $p,q \ge 2$ and $1 \le p \le 2 \le q$, as $1 \le p,q \le 2$ is mathematically equivalent to $p,q \ge 2$. In the region $1 \le q \le 2 \le p$, the matrix $p \rightarrow q$ norm can be approximated up to a constant factor of at most $\pi /2$, this is often called Nesterov's $\pi/2$ Theorem, and the region $1 \le q \le 2 \le p$ is often called the Nesterov region \cite{nesterov1998semidefinite}. When the underlying matrix is non-negative, the matrix $p \rightarrow q$ norm can be computed exactly when $p \ge q \ge 1$ (see \cite{steinberg2005computation,bhaskara2011approximating}), a setting with application to oblivious routing in graph theory \cite{englert2009oblivious}. 

\subsection{Hardness Results} Outside of these special cases, the majority of results for matrix norms consists of algorithmic lower bounds. In \cite{rohn2000computing}, Rohn showed that computing the $\infty \rightarrow 1$ norm is NP-hard, a result later expanded upon by Steinberg, who showed that computing the $p \rightarrow q$ norm for $p >q$ is NP-hard \cite{steinberg2005computation}. Later, Hendrickx and Olshevsky proved that computing the $p \rightarrow p$ norm for any rational $p \ne 1,2$ is NP-hard, and that computing the $\infty \rightarrow q$ norm to an arbitrary constant for any rational $q$ is NP-hard \cite{hendrickx2010matrix}. Bhaskara and Vijayaraghavan proved that it is NP-hard to compute the $p\rightarrow q$ norm up to any constant factor for $2 < q \le p$, and, under a reasonable assumption from computational complexity, show that this norm cannot be approximated within a factor of $2^{\log^{1-\epsilon}(n)}$ for any constant $\epsilon >0$ \cite{bhaskara2011approximating}. Barak et al. treated the $2\rightarrow q$ norm problem, and showed that the $2\rightarrow 4$ norm is NP-hard to approximate up to inverse-polynomial (in dimension) precision, and, assuming the exponential time hypothesis, better than a factor of $2^{O(\log ^{1/2 - \epsilon} n)}$ \cite{barak2012hypercontractivity}. Finally, Bhattiprolu et al. showed, under a reasonable assumption from computational complexity, that it is NP-hard to approximate the $p \rightarrow q$ norm for $2 < p < q$ within a factor of $2^{O(\log^{1-\epsilon} n)}$ \cite{bhattiprolu2023inapproximability}.

\subsection{Existing Algorithms} In terms of methods, in 1974 Boyd proposed a power method type heuristic for computing the matrix $p \rightarrow q$ norm, with theoretical analysis for the case of non-negative matrices, but without rigorous guarantees for arbitrary matrices \cite{boyd1974power}. Higham later produced a modified power method using a well-chosen initial guess that performed well in practice for the $p \rightarrow p$ norm and had the theoretical approximation guarantee of $n^{1-1/p}$ for $p \ge 2$ \cite{higham1992estimating}. Steinberg, by approximating norms in the Nesterov region $1 \le q \le 2 \le p$ via semidefinite programming and using interpolation, produced a $(\pi/2)^{\frac{1}{q}} m^{\frac{q-2}{q^2}}$ approximation approximation algorithm for $p,q \ge 2$ and a $m^{\frac{2(q-2)(p-1)}{p q^2}} \, n^{\frac{(p-1)(2-p)}{p^2}}$ algorithm for $p\le 2 \le q$, leading to a worst-case approximation of $O(\max\{m,n\}^{25/128})$ for an arbitrary $p$ and $q$. For the $p \rightarrow p$ norm, Steinberg's algorithm has a stronger theoretical guarantee than that of Higham (e.g., $m^{1/8}$ vs $m^{1/4}$ for $\|A\|_{4 \rightarrow 4}$), but this theoretical gap is offset by the practical difference in efficiency of their two algorithms, as Steinberg's estimate relies on the solution of a semidefinite program. See \cite{jiang2020faster} for a summary of the existing  literature and best-known complexity for cutting plane and interior point techniques for approximately solving a semidefinite program. As previously mentioned, the $2\rightarrow q$ norm is of particular interest in theoretical computer science. One interesting area of research for $2\rightarrow q$ norms is in producing algorithms with bounded additive error. Barak et al. produced a $2^{O(\log^2(n) \epsilon^{-2})}$ time algorithm that computes $\|A\|^4_{2\rightarrow 4}$ up to additive error $\epsilon \|A\|^2_{2 \rightarrow 2} \|A\|^2_{2 \rightarrow \infty}$ \cite{barak2012hypercontractivity}, a result later improved and extended to arbitrary $q$ by Brand\~ao and Harrow \cite{brandao2015estimating}.

\subsection{Our Results} In this work, we improve upon existing algorithms for computing $\|A\|_{p \rightarrow q}$. We do so by producing superior estimates for the matrix $2 \rightarrow q$ norm (Lemma \ref{lm:2toq}) and extending these results using interpolation. Our simplest result for estimating $\|A\|_{2\to q}$ is recorded in the following theorem. %Our improved estimate for $\|A\|_{2\to q}$ is the following. The estimate for In particular, we prove the following.

\begin{theorem}\label{thm:simplep=2}
Let $A \in \mathbb{C}^{m \times n}$, $r_1,...,r_m \in \mathbb{C}^n$ denote the conjugate transpose of the $m$ rows of $A$, $v \in \mathbb{C}^n$ denote a right singular vector corresponding to the $2$-norm of $A$, and $S = \{v,r_1,...,r_m\}$. Then
$$\max_{\substack{x \in S  \\ x \ne 0} } \, \frac{\|A x\|_q}{\|x\|_2} \ge  (2m^{\frac{1}{q}})^{-\frac{q-2}{2(q-1)}}  \|A\|_{2 \rightarrow q}\qquad \text{for all} \quad  q \ge 2. $$ 

\end{theorem}
In \textsection \ref{exs}, we show that Theorem \ref{thm:simplep=2} is the best approximation possible with the given data. We also give examples to demonstrate the significantly improved bounds from Theorem \ref{thm:simplep=2} compared to Steinberg's estimate using interpolation between the $2\to 2$ and $2\to\infty$ norms \cite{steinberg2005computation}. 
See Figure \ref{fig:plots} for a comparison of the approximation exponents for Steinberg's interpolation approach and for Theorem \ref{thm:simple}. As an example, interpolation gives an $\sqrt[3]{2} \, m^{1/12}$ approximation for the popular $2 \rightarrow 4$ norm, while interpolation provides only a $m^{1/8}$ approximation. 

By also incorporating the $1\to q$ norm and interpolation, we use Theorem \ref{thm:simplep=2} to improve $p\to q$ norm estimates for a range of $p$ and $q$ exponents, which leads to the following generalization of Theorem \ref{thm:simplep=2}.  
\begin{theorem}\label{thm:simple}
Let $A \in \mathbb{C}^{m \times n}$, $r_1,...,r_m \in \mathbb{C}^n$ denote the conjugate transpose of the $m$ rows of $A$, $v \in \mathbb{C}^n$ denote a right singular vector corresponding to the $2$-norm of $A$, and $y \in \mathbb{C}^n$ achieve the $1 \rightarrow q$ norm of $A$, and $S = \{v,r_1,...,r_m\}$. Then
$$\max_{\substack{x \in S \cup \{y\} \\ x \ne 0} } \, \frac{\|A x\|_q}{\|x\|_p} \ge \left[ (2m^{\frac{1}{q}})^\frac{q-2}{q-1} n^{\frac{2-p}{p}}\right]^{-\frac{p-1}{p}}  \|A\|_{p \rightarrow q}$$ 
for all
$\big\{ (p,q) \, | \, p \in [1,\sqrt{2}], \, q \in \big[\frac{2}{2-p},\frac{p}{p-1}\big]\big\}$ and $\big\{ (p,q) \, | \, p \in [\sqrt{2},2], \, q \in \big[\frac{p}{p-1},\frac{2}{2-p}\big]\big\}$, and
$$  \max_{\substack{x \in S \\ x \ne 0} } \, \frac{\|A x\|_q}{\|x\|_p} \ge \left[ (2m^{\frac{1}{q}})^\frac{q-2}{q-1} n^{\frac{p-2}{p}}\right]^{-\frac{1}{2}}  \|A\|_{p \rightarrow q} $$
for all $p,q \in [2,\infty)$.
\end{theorem}

We note that improved estimates for the regions not directly addressed in Theorem \ref{thm:simple} follow from duality, as $p,q \in [1,2]$, $\big\{ (p,q) \, | \, q \in \big[\frac{\sqrt{2}}{\sqrt{2}-1},\infty\big), \, p \in \big[\frac{q}{q-1},\frac{2(q-1)}{q}\big]\big\}$, and $\big\{ (p,q) \, | \, q  \in \big[2,\frac{\sqrt{2}}{\sqrt{2}-1}\big], \, p \in \big[\frac{2(q-1)}{q},\frac{q}{q-1}\big]\big\}$ are dual to $p,q \in [2,\infty)$, $\big\{ (p,q) \, | \, p \in [1,\sqrt{2}], \, q \in \big[\frac{2}{2-p},\frac{p}{p-1}\big]\big\}$, and $\big\{ (p,q) \, | \, p \in [\sqrt{2},2], \, q \in \big[\frac{p}{p-1},\frac{2}{2-p}\big]\big\}$, respectively. For the region $p,q \in [2,\infty)$ (and, again, by duality, $p,q \in [1,2]$), the worst-case estimates for our simple estimator outperforms semidefinite programming-based interpolation estimates for square matrices when $ p < 2 q^2(q-1)/(q^3-2q^2+4q-4)$. We give a visual representation of these details, as well as a contour plot of the best approximation exponent with respect to $p$ and $q$ in Figure \ref{fig:plots}. Theorem \ref{thm:simple} is a simplified version of our main result that makes use of the worst-case Lipschitz bounds for norms. We prove a tighter version of this result in Theorem \ref{thm:complicated}, which has the same worst-case behavior, but improved bounds for the majority of matrices $A$. Overall, Theorem \ref{thm:simple} implies a global $O\big(\max\{m,n\}^{3-2\sqrt{2}}\big)$ approximation algorithm, as
$$\max \bigg\{ \frac{(p-1)(2-p)}{p^2} + \frac{(p-1)(q-2)}{pq(q-1)} \, \bigg| \, 1\le p \le 2 \le q, \, \frac{1}{p} + \frac{1}{q} \le 1\bigg\} = 3 - 2 \sqrt{2},$$
achieved by $p = \sqrt{2}$ and $q = \frac{\sqrt{2}}{\sqrt{2}-1}$. We note that the approximations provided by Theorem \ref{thm:simple} are not simply asymptotic theoretical bounds, but concrete estimates that can be used in practice. In Table \ref{tab:numbers}, we provide the worst-case approximation ratio for the $2\rightarrow 4$ norm and any $p \to q$ norm for matrices of dimension $10^3$ to $10^9$. Theorem \ref{thm:simple} gives an approximation for any million by million matrix in any $p \rightarrow q$ norm within a factor of $11.081$, and within a factor of $3.985$ for the $2 \rightarrow 4$ norm. This illustrates that, while for sufficiently large matrices there is no hope for a polylogarithmic approximation to the $p \rightarrow q$ norm, this asymptotic algorithmic lower bound does not apply to the sizes of matrices typically seen in practice.

\begin{figure}[t]
    \centering
    \subfigure[Best $p\rightarrow q$ norm approximation by region]{\includegraphics[width=.468\linewidth]{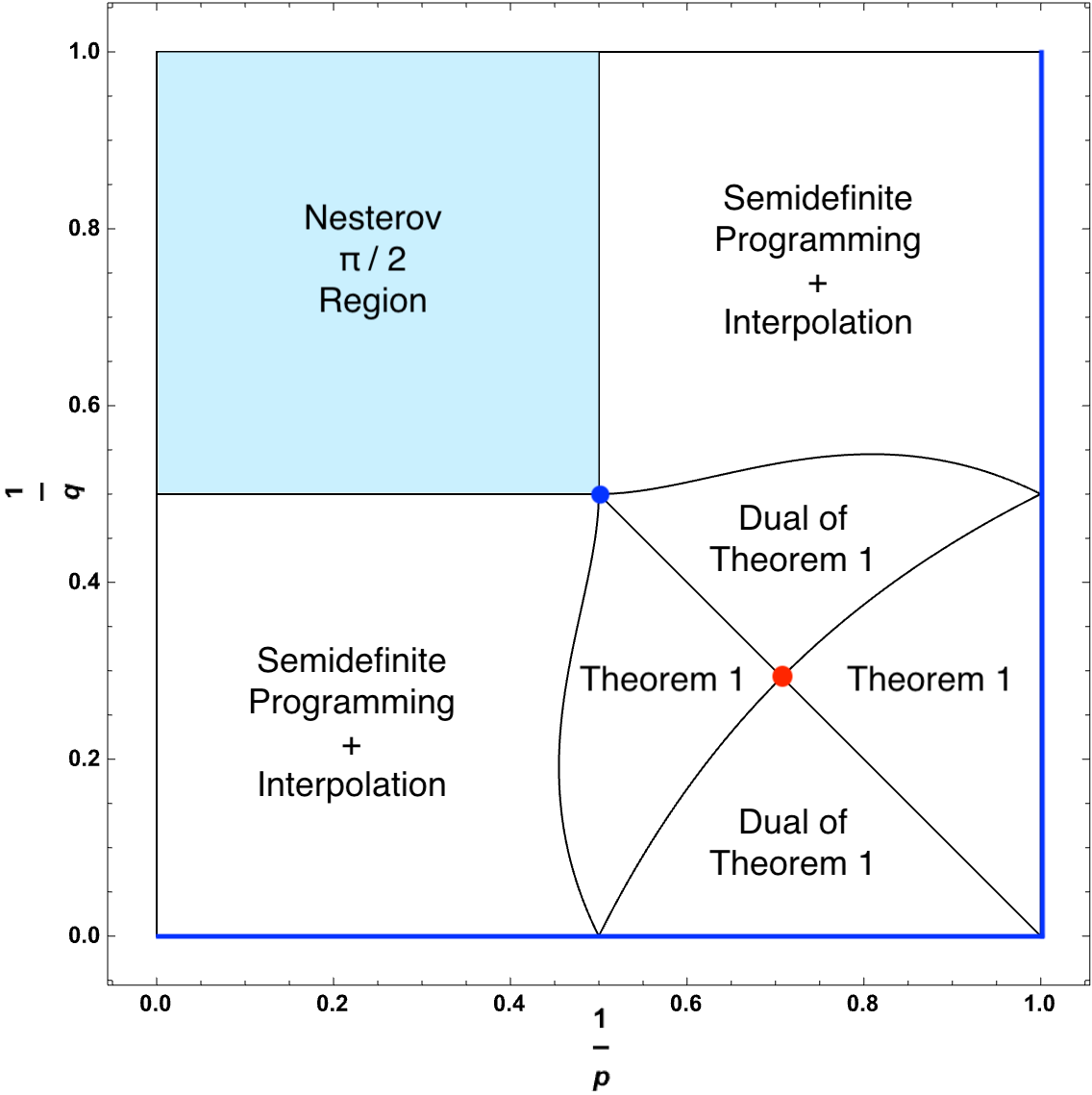}} \qquad
    \subfigure[Exponent of approximation for $2\rightarrow q$ norm]{\includegraphics[width=.468\linewidth]{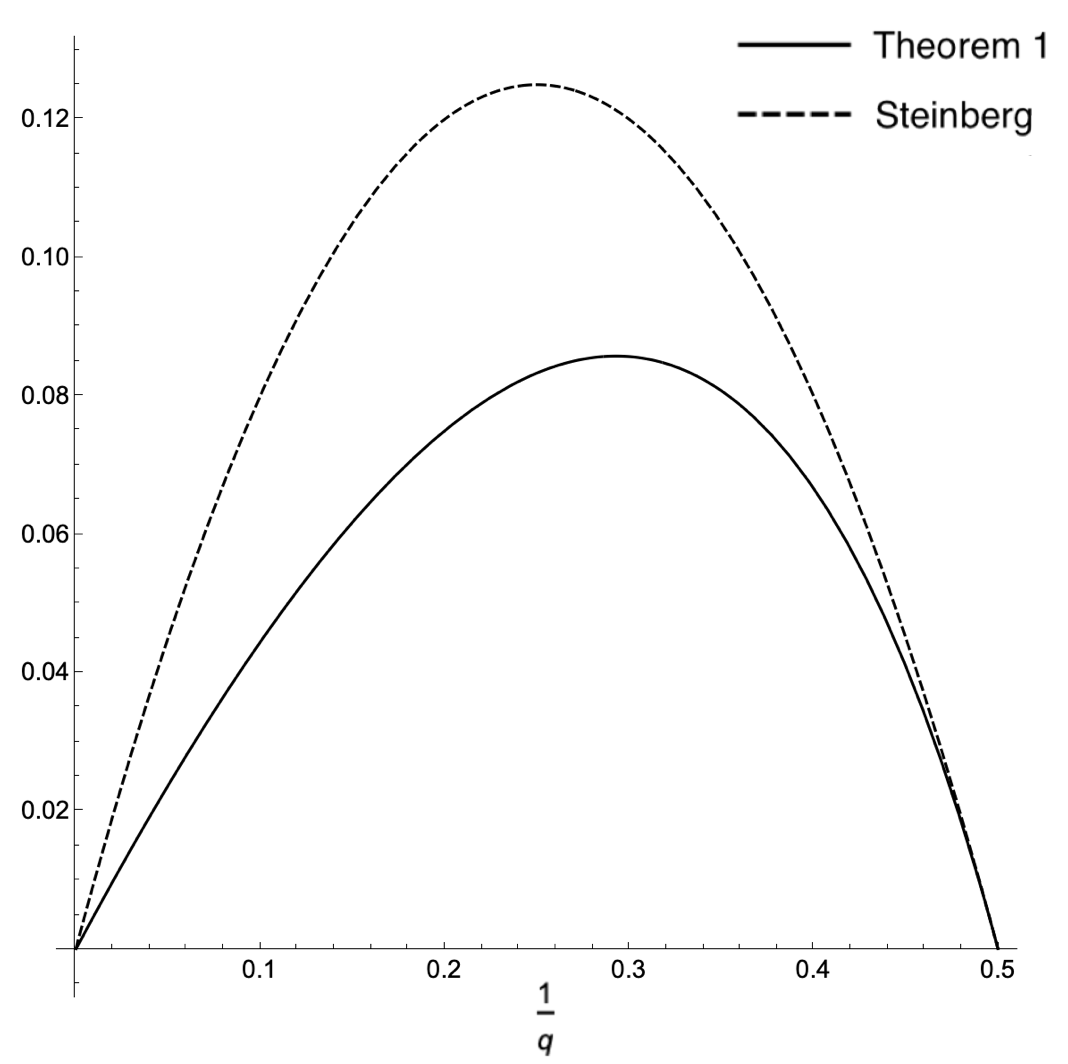}} 
    \caption{ Figure (a) is a plot of the square $(1/p,1/q) \in [0,1]^2$. Efficiently computable $p,q$ are in dark blue and the Nesterov region computable to a constant factor is in light blue. We note the regions of $p,q$ for which Theorem \ref{thm:simple} improves upon the best-known approximation. The worst approximation for any $p,q$ is given by $p = \sqrt{2}$, $q = \sqrt{2}/(\sqrt{2}-1)$, marked in red. Figure (b) plots the worst-case approximation exponent (e.g., the value $\alpha$ in a $C m^\alpha$ approximation) for Steinberg's interpolation approach and Theorem \ref{thm:simple}. For the popular $2 \rightarrow 4$ norm, interpolation achieves an exponent of $1/8$, while Theorem \ref{thm:simple} achieves an exponent of $1/12$, a $33.\bar{3}\%$ decrease in the exponent.}
     \label{fig:plots}
\end{figure}

\begin{table}[t]
\centering
\begin{tabular}{|l || r | r | r | r| r| r|r|}
 \hline $\max\{m,n\} 
= $ & $10^3$ & $10^4$ & $10^5$ & $10^6$ & $10^7$ & $10^8$ & $10^9$ \\
\hline \hline
$2 \rightarrow 4$ norm &\, 2.241 & \, 2.715 & \,3.289 & \, 3.985& \, 4.827 & \, 5.849 & \,  7.086 \\ \hline
$p \rightarrow q$ norm & 3.388 & 5.029& 7.465 & 11.081 & 16.449 & 24.419 & 36.248 \\\hline 
\end{tabular}
\vspace{4 mm}
\caption{The worst-case approximation ratio for the $2\rightarrow 4$ norm and the $p \rightarrow q$ norm (achieved by $p = \sqrt{2}$, $q = \sqrt{2}/(\sqrt{2}-1)$) for matrices of varying sizes of $m$ and $n$. Despite the non-existence of an efficient polylogarithmic approximation algorithm for the matrix $p \rightarrow q$ norm problem, the exponents in the polynomial factors of Theorem \ref{thm:simple} grow quite slowly. Indeed, the $p\rightarrow q$ norm of a million by million sized matrix can be quickly approximated within a factor of $\approx 11$, while $\ln(10^6) \approx 13.8$. In the case of the $2\rightarrow 4$ norm, a billion by billion sized matrix can be approximated up to a factor of $\approx 7$, while $\ln(10^{9}) \approx 20.7$.}
\label{tab:numbers}
\end{table}

The remainder of the paper is as follows. In Section \ref{sec:2toq}, we study the matrix $2\rightarrow q$ norm problem. We produce improved results by moving beyond standard interpolation bounds and techniques and exploiting the information contained in each row of a matrix. We illustrate the tightness of our bounds using a variety of known matrices. In Section \ref{sec:extend}, we extend our results for $2\rightarrow q$ norms to arbitrary $p \rightarrow q$ norms, producing both a proof of Theorem \ref{thm:simple} and inspiring a practical algorithm for estimating $p \rightarrow q$ norms. Finally, we perform numerical experiments illustrating the practical efficiency of the theoretical results of this work.

\noindent\textbf{Acknowledgements} LG is supported by a Simons Investigator award. DM is supported by the National Science Foundation under Award No. 2103249.

\section{Estimating the $2 \rightarrow q$ Norm}\label{sec:2toq}

The approximation algorithm of Steinberg approximates the $2 \to q$ norm, $q \ge 2$, of a matrix by computing the $2 \to 2$ and $2 \to \infty$ norms and using interpolation. This estimate depends on two vectors, a right singular vector corresponding to the largest singular value and the conjugate transpose of a row with maximal $2$-norm. In this section, we prove that making use of all rows of the matrix leads to an improved approximation algorithm. We present the following lemma.

\begin{lemma}\label{lm:2toq}
Let $A \in \mathbb{C}^{m \times n}$, $r_1,...,r_m \in \mathbb{C}^n$ denote the conjugate transpose of the $m$ rows of $A$, $v \in \mathbb{C}^n$ denote a right singular vector corresponding to the $2$-norm of $A$, and $S = \{v,r_1,...,r_m\}$. Then
$$\max_{\substack{x \in S \\ x \ne 0} } \, \frac{\|A x\|_q}{\|x\|_2} \ge  \left[\frac{1}{2\|A\|_{2\rightarrow \infty}}\max_{\substack{x \in S \\ x \ne 0}} \frac{\|Ax\|_q}{\|x\|_2 } \right]^{\frac{q-2}{2(q-1)}} \left[\frac{\|Av\|_q}{\|Av\|_2} \right]^{\frac{1}{q-1}} \|A\|_{2 \rightarrow q} \qquad \text{for all} \quad  q \ge 2.$$ 
\end{lemma}

\begin{proof}
 Let $w \in \mathbb{C}^n$, $\|w \|_2 =1$, satisfy $\|Aw \|_q = \|A\|_{2 \rightarrow q}$, and consider the orthogonal projection of $w$ onto $r_1,...,r_m$. If, for some $\lambda >0$, $|\langle w, r_i \rangle | \le \lambda \|A\|_{2 \rightarrow \infty}$ for all $i = 1,...,m$, then 
$$\|A w\|_q^q = \sum_{i=1}^m |\langle w,r_i \rangle|^q \le \big[\lambda \|A\|_{2 \rightarrow \infty}\big]^{q-2} \sum_{i =1}^m|\langle w,r_i \rangle|^2 \le \big[\lambda \|A\|_{2\rightarrow\infty}\big]^{q-2}\|A\|_{2\rightarrow2}^2,$$
 \begin{align*}
     \|A\|_{2 \rightarrow q} &\le \lambda^{\frac{q-2}{q}} \|A\|_{2 \rightarrow \infty}^{\frac{q-2}{q}}\|A\|_{2 \rightarrow 2}^{\frac{2}{q}}\\
     &= \lambda^{\frac{q-2}{q}} \bigg[ \min_{\substack{x \in S \\ x \ne 0}} \frac{\|A x\|_{q}}{\|x\|_2}  \, \frac{\|x\|_{2} \|A\|_{2\rightarrow \infty}}{\|Ax\|_q}\bigg]^{\frac{q-2}{q}} \bigg[ \frac{\|Av\|_q}{\|v\|_2} \, \frac{\|A v\|_2}{\|Av\|_q} \bigg]^{\frac{2}{q}} \\
     &\le \lambda^{\frac{q-2}{q}} \bigg[\frac{\|Av\|_2}{\|Av\|_q} \bigg]^{\frac{2}{q}} \bigg[\min_{\substack{x \in S \\ x \ne 0}} \frac{\|x\|_2\|A\|_{2\rightarrow \infty}}{\|A x\|_q} \bigg]^{\frac{q-2}{q}}  \max_{\substack{x \in S \\ x \ne 0}} \frac{\|Ax\|_q}{\|x\|_2}. 
 \end{align*}
If there exists some $r_k$ such that $|\langle w,r_k \rangle|> \lambda \|A\|_{2 \rightarrow \infty}$, then we can upper bound $\|A\|_{2 \rightarrow q}$ by considering the orthogonal projection of $w$ onto $r_k$. We have
$$ \bigg\| w - \frac{\langle w,r_k\rangle}{\|r_k\|_2^2} r_k \bigg\|_2^2 = \| w\|_2^2  - \frac{|\langle w,r_k \rangle|^2}{\|r_k\|_2^2} =  1  - \frac{|\langle w,r_k \rangle|^2}{\|r_k\|_2^2}$$
and
\begin{align*}
     \|A w \|_q &\le \bigg\| A \bigg[\frac{\langle w,r_k \rangle}{\|r_k\|_2^2} \, r_k \bigg]\bigg\|_q + \bigg\|A \bigg[w - \frac{\langle w,r_k \rangle}{\|r_k\|_2^2} \, r_k \bigg] \bigg\|_q \\
    &\le  \frac{|\langle w,r_k \rangle|}{\|r_k\|_2^2}\|A r_k \|_q + \|A\|_{2 \rightarrow q} \bigg\| w - \frac{\langle w,r_k \rangle}{\|r_k\|_2^2} r_k \bigg\|_2\\
    &=  \frac{|\langle w,r_k \rangle|}{\|r_k\|_2^2}\|A r_k \|_q + \|A\|_{2 \rightarrow q} \bigg[1  - \frac{|\langle w,r_k \rangle|^2}{\|r_k\|_2^2} \bigg]^{1/2}\\
    &\le \frac{|\langle w,r_k \rangle|}{\|r_k\|_2^2}\|A r_k \|_q + \|A\|_{2 \rightarrow q} \bigg[1  - \frac{|\langle w,r_k \rangle|^2}{2\|r_k\|_2^2} \bigg],
\end{align*}
implying that
$$ \frac{\|Ar_k\|_q}{\|r_k\|_2}  \ge \frac{|\langle w,r_k \rangle|}{2\|r_k\|_2} \|A\|_{2 \rightarrow q} > \frac{\lambda}{2} \, \frac{\|A\|_{2 \rightarrow \infty}}{\|r_k\|_2} \|A\|_{2 \rightarrow q}  \ge \frac{\lambda}{2} \|A\|_{2 \rightarrow q}.$$
Altogether, we have that
$$\max_{\substack{x \in S \\ x \ne 0} } \, \frac{\|A x\|_q}{\|x\|_2} \ge \min \Bigg\{ \frac{\lambda}{2},\frac{1}{\lambda^{\frac{q-2}{q}}} \bigg[\frac{\|Av\|_q}{\|Av\|_2} \bigg]^{\frac{2}{q}} \bigg[\max_{\substack{x \in S \\ x \ne 0}} \frac{\|A x\|_q}{\|x\|_2\|A\|_{2\rightarrow \infty}} \bigg]^{\frac{q-2}{q}}   \Bigg\} \, \|A\|_{2 \rightarrow q}.$$
The quantity $\displaystyle{\min \Bigg\{ \frac{\lambda}{2},\frac{1}{\lambda^{\frac{q-2}{q}}} \bigg[\frac{\|Av\|_q}{\|Av\|_2} \bigg]^{\frac{2}{q}} \bigg[\max_{\substack{x \in S \\ x \ne 0}} \frac{\|A x\|_q}{\|x\|_2\|A\|_{2\rightarrow \infty}} \bigg]^{\frac{q-2}{q}}  \Bigg\}}$ is maximized when $$\lambda = 2^{\frac{q}{2(q-1)}}\left[\frac{\|Av\|_q}{\|Av\|_2} \right]^{\frac{1}{q-1}} \left[\max_{\substack{x \in S \\ x \ne 0}} \frac{\|A x\|_q}{\|x\|_2\|A\|_{2\rightarrow \infty}}  \right]^{\frac{q-2}{2(q-1)}},$$ implying our desired result.
\end{proof}
We note that the bound of Lemma \ref{lm:2toq} has a constant term approaching $\sqrt{2}$ as $q$ tends to infinity. This lemma is clearly not tight for $q = \infty$, giving an estimate of $\sqrt{2}$ instead of $1$. However, this is not an artifact of the bounding procedure, but the technique itself, as, for any fixed $q<\infty$, this factor is unavoidable for matrices whose optimal $\lambda$ is a function of $m$.

Theorem \ref{thm:simplep=2} follows as a corollary to Lemma \ref{lm:2toq}. Although the estimate in Theorem \ref{thm:simplep=2} is simpler, for a particular matrix $A$, the bounds from Lemma \ref{lm:2toq} can be tighter. 
\begin{proof}[Proof of Theorem \ref{thm:simplep=2}] First observe that  
$$\frac{1}{2\|A\|_{2\rightarrow \infty}}\max_{\substack{x \in S \\ x \ne 0}} \frac{\|Ax\|_q}{\|x\|_2} \ge \frac{\|Ar_j\|_q}{2\|Ar_j\|_\infty}.$$
It remains to note that $\frac{\|Ar_j\|_q}{2\|Ar_j\|_\infty}\ge \frac{1}{2}$ and $\frac{\|Av\|_q}{\|Av\|_2}\ge {m^{-\frac{q-2}{2q}}}$, which follows from H\"{o}lder's inequality. 
\end{proof}

\subsection{Examples \label{exs}}
We give examples comparing the bounds from Lemma \ref{lm:2toq} with the interpolation bounds. We show that Lemma \ref{lm:2toq} is essentially the best approximation possible given our inputs. We also demonstrate the dependence of Lemma \ref{lm:2toq} on the choice of witness $v$ of the spectral norm. 
\subsubsection{Tightness of Lemma \ref{lm:2toq} and comparison with interpolation} We begin by recalling the classical interpolation estimates for $\|A\|_{2\to q}$ which are recorded in \cite{steinberg2005computation}. Since $2\le q\le \infty$, the matrix norm $\|A\|_{2\to q}$ is controlled by $\|A\|_{2\to2}$ and $\|A\|_{2,\infty}$ by interpolation. Indeed, if $w\in\C^n$ satisfies $\|A\|_{2\to q}=\|Aw\|_{q}$ and $\|w\|_2=1$, then 
\[  \|Aw\|_{q}^q=\sum_{i=1}^m  |\langle w,r_i\rangle|^q\le \max_{1\le k\le m}|\langle w,r_{k}\rangle|^{q-2}\sum_{i=1}^m  |\langle w,r_i\rangle|^2\le \|A\|_{2\to \infty}^{q-2}\|A\|_{2\to2}^2,\]
where we use the same notation as in Lemma \ref{lm:2toq}. Let $r_j\in\R^n$ be a row of $A$ with maximal $\|r_j\|_2$ and let $v \in\R^n$ satisfy $\|A\|_{2\to 2}=\frac{\|Av\|_2}{\|v\|_2}$. Then the interpolation inequality is 
%\[ \] above inequality implies We also have the lower bounds $\|A\|_{2,\infty}\le \|A\|_{2\to q}$ and, using H\"{o}lder's inequality, 
%\[ \|A\|_{2\to2}^2=\frac{\|Av\|_2^2}{\|v\|_2^2}\le m^{\frac{q-2}%{q}}\frac{\|Av\|_q^{2}}{\|v\|_2^2} \le m^{\frac{q-2}{q}}\|A\|_{2\to q}^2. \]
%All together, this yields the interpolation inequality 
\begin{equation}\label{eqn:int} 
\max\Big(\frac{\|Ar_j\|_q}{\|r_j\|_2},\frac{\|Av\|_q}{\|v\|_2}\Big)\le \|A\|_{2\to q}\le \left[\frac{\|Ar_j\|_\infty}{\|r_j\|_2}\right]^{\frac{q-2}{q}}\left[\frac{\|Av\|_{2}}{\|v\|_2}\right]^{\frac{2}{q}} . \end{equation}
Since it is straightforward to compute $r_j$ and to approximate $v$, \eqref{eqn:int} gives a polynomial time approximation of $\|A\|_{2\to q}$ with a multiplicative error bounded by $\left[\frac{\|Ar_j\|_\infty}{\|Ar_j\|_q}\right]^{\frac{q-2}{q}}\left[\frac{\|Av\|_{2}}{\|Av\|_q}\right]^{\frac{2}{q}}$. %The only calculations involved in the interpolation inequality are
%\[ \|A\|_{2\to2}=\frac{\|Av\|_2}{\|v\|_2}\qquad\text{and}\qquad \|A\|_{2,\infty}=\max_{1\le i\le m}\|r_i\|_2 . \]

The bounds in Lemma \ref{lm:2toq} involve more calculations: we require $\|Av\|_q/\|v\|_2$ and $\|Ar_i\|_q/\|r_i\|_2$ for each $i=1,\ldots,m$. Using Lemma \ref{lm:2toq}, we approximate $\|A\|_{2,q}$ by 
\begin{equation}\label{eqn:lm} \max_{\substack{x \in S \\ x \ne 0} } \, \frac{\|A x\|_q}{\|x\|_2} \le \|A\|_{2 \rightarrow q} \le \left[\max_{\substack{x \in S \\ x \ne 0}} \frac{\|Ax\|_q}{2\|A\|_{2\rightarrow \infty}\|x\|_2 } \right]^{-\frac{q-2}{2(q-1)}} \left[\frac{\|Av\|_q}{\|Av\|_2} \right]^{-\frac{1}{q-1}}\max_{\substack{x \in S \\ x \ne 0} } \, \frac{\|A x\|_q}{\|x\|_2}  . \end{equation}

We will give examples of $m\times m$ matrices $B$ and $D$ which have identical data as inputs for both \eqref{eqn:int} and \eqref{eqn:lm}. For these examples, the range for the $2\to q$ norms given by interpolation is a factor of $m^{\frac{(q-2)^2}{2q^2(q-1)}}$ larger than the range given by Lemma \ref{lm:2toq}. Furthermore, the $2\to q$ norms of $B$ and $D$ take values in the opposite ends of the range given by Lemma \ref{lm:2toq}, which demonstrates that the approximation is the best possible with the given data. In this section, the notation $\sim$ and $\lesssim$ means within a factor of $10$.

Let $n=m+2$.  Let $e_i\in\R^n$ be the standard basis vectors with a $1$ in the $i$th coordinate and $0$ otherwise. Define the rows of the $m\times n$ matrix $B$ by $r_i^B=(1-m^{-\frac{2}{q}})^{\frac{1}{2}} e_i+m^{-\frac{1}{q}} e_n$ for $i=1,\ldots,m$. Thus $B$ has the form
\[ \left[
    \begin{array}{c;{2pt/2pt}c}
    [1-m^{-\frac{2}{q}}]^{1/2}\mbox{\large $I_m$} &  \begin{matrix}
    0&m^{-\frac{1}{q}}\\
    \vdots&\vdots\\
    0&m^{-\frac{1}{q}} \\
\end{matrix} 
    \end{array}
    \right]
  \]
in which $I_m$ is the $m\times m$ identity.

Define the rows of the $m\times n$ matrix $D$ by $r_i^D=m^{-\frac{1}{q}}e_1+m^{-\frac{q-2}{2q(q-1)}}e_2+[1-m^{-\frac{2}{q}}-m^{-\frac{q-2}{q(q-1)}}]^{1/2}e_{i+2}$ if $1\le i\le \lfloor m^{\frac{q-2}{q-1}}\rfloor$ and $r_i^D=m^{-\frac{1}{q}}e_1+[1-m^{-\frac{2}{q}}]^{1/2}e_{i+2}$ if $\lfloor m^{\frac{q-2}{q-1}}\rfloor+1\le i\le m$. This gives the  matrix $D$ the form
\[ \left[
    \begin{array}{c;{2pt/2pt}c;{2pt/2pt}c}
    \begin{matrix}
    m^{-\frac{1}{q}}&m^{-\frac{q-2}{2q(q-1)}}\\
    \vdots&\vdots\\
    m^{-\frac{1}{q}}&m^{-\frac{q-2}{2q(q-1)}} \\
\end{matrix}
 & [1-m^{-\frac{2}{q}}-m^{-\frac{q-2}{q(q-1)}}]^{1/2}\mbox{\large $I_{\lfloor m^{\frac{q-2}{q-1}}\rfloor}$} & {0} \\ \hdashline[2pt/2pt]
    \begin{matrix} 
    m^{-\frac{1}{q}} {\quad \;\;\,}&0 {\quad \; \;} \\
    \vdots {\quad \;\;\,}&\vdots{\quad \; \;} \\
    m^{-\frac{1}{q}} {\quad \;\;\,} &0 {\quad \;\;}\\ \end{matrix} & \large{0} & [1-m^{-\frac{2}{q}}]^{1/2} \mbox{\large $I_{m-\lfloor m^{\frac{q-2}{q-1}}\rfloor}$}
    \end{array}
    \right]. 
  \]
The data that is the same for both $B$ and $D$ is the following:   
\begin{enumerate}
    \item $\|r_i^B\|_2=\|r_i^D\|_2=1$ for all $i=1,\ldots,m$,
    \item $\|B\|_{2\to\infty}=\|D\|_{2\to\infty}=1$ and $\|B\|_{2\to2}\sim\|D\|_{2\to 2}\sim m^{\frac{q-2}{2q}}$,
    \item $\|Br_i^B\|_q\sim\|Br_i^B\|_\infty\sim 1$ and $\|Dr_i^D\|_q\sim\|Dr_i^D\|_\infty\sim 1$ for all $i=1,\ldots,m$. 
\end{enumerate}
Approximate witnesses of the spectral norms of $B$ and $D$ are $v_B=m^{\frac{1}{q}-1}\sum_{j=1}^mr_j$ and $v_D=e_1$, respectively. For these choices of approximate witnesses, write $S_B=\{r_1^B,\ldots,r_m^B,v_B\}$ and $S_D=\{r_1^D,\ldots,r_m^D,v_D\}$. We have
\[ \max_{\substack{x\in S_B\\ x\not=0}}\frac{\|Bx\|_q}{\|x\|_2}\sim \max_{\substack{x\in S_D\\ x\not=0}}\frac{\|Dx\|_q}{\|x\|_2}\sim 1. \]
We also note that $\|Bv_B\|_2\sim\|Dv_D\|_2\sim m^{\frac{1}{2}-\frac{1}{q}}$ and $\|Bv_B\|_q\sim \|Dv_D\|_q\sim 1$. 

Plugging all of this data into \eqref{eqn:int} and \eqref{eqn:lm} leads to the following ranges for the $2\to q$ norms:
\begin{align*}
    \|B\|_{2\to q},\|D\|_{2,q}\in [\frac{1}{10},10m^{\frac{q-2}{q^2}}]&\qquad \text{using interpolation, and } \\
    \|B\|_{2\to q},\|D\|_{2,q}\in [\frac{1}{10},10 m^{\frac{q-2}{2q(q-1)}}]&\qquad \text{using Lemma \ref{lm:2toq}}. 
\end{align*}
Furthermore, $\|B\|_{2\to q}\sim \|Be_1\|_q\sim 1$, demonstrating the sharpness of the lower bound from Lemma \ref{lm:2toq}. For $D$, we have $\|D\|_{2\to q}\sim \|De_2\|_q\sim m^{\frac{q-2}{2q(q-1)}}$, which shows that the upper bound from Lemma \ref{lm:2toq} is sharp. 

Finally, we note the dependence of \eqref{eqn:lm} on the choice of witness for the spectral norm. The vector $w_D=e_2$ also approximately achieves the spectral norm $\|D\|_{2\to 2}\sim m^{\frac{q-2}{2q}}$. Using $w_D$ in place of $v_D$ in \eqref{eqn:lm}, we have
\[  m^{\frac{q-2}{2q(q-1)}}\lesssim \|D\|_{2\to q}\lesssim m^{\frac{(q-2)^2}{4q(q-1)^2}}m^{\frac{q-2}{2q(q-1)}},\]
which is sharp in the lower bound.

\section{Improved Matrix $p \rightarrow q$ Estimation}\label{sec:extend}

\subsection{From $2\rightarrow q$ to $p \rightarrow q$}

Using Lemma \ref{lm:2toq} and interpolation, we obtain improved estimates for the $p \rightarrow q$ norm. We have the following theorem.

\begin{theorem}\label{thm:complicated}
Let $A \in \mathbb{C}^{m \times n}$, $r_1,...,r_m \in \mathbb{C}^n$ denote the conjugate transpose of the $m$ rows of $A$, $r_j$ be the vector with the largest 2-norm, $v \in \mathbb{C}^n$ denote a right singular vector corresponding to the $2$-norm of $A$, $y \in \mathbb{C}^n$ achieve the $1 \rightarrow q$ norm of $A$, and $S = \{v,r_1,...,r_m\}$. Then
$$\max_{\substack{x \in S \cup \{y\} \\ x \ne 0} } \, \frac{\|A x\|_q}{\|x\|_p} \ge  \left[\frac{\|y\|_1}{\|y\|_p} \right]^{\frac{2-p}{p}} \bigg( \left[\frac{\|Ar_j\|_q}{2\|Ar_j\|_\infty} \right]^{\frac{q-2}{2(q-1)}} \left[\frac{\|Av\|_q}{\|Av\|_2} \right]^{\frac{1}{q-1}} \min_{\substack{x \in S  \\ x \ne 0}} \frac{\|x\|_2}{\|x\|_p} \bigg)^{\frac{2(p-1)}{p}}\|A\|_{p \rightarrow q} $$ 
for all $1 \le p \le 2 \le q$, and 
$$\max_{\substack{x \in S  \\ x \ne 0} } \, \frac{\|A x\|_q}{\|x\|_p} \ge n^{\frac{2-p}{2p}} \bigg( \left[\frac{\|Ar_j\|_q}{2\|Ar_j\|_\infty} \right]^{\frac{q-2}{2(q-1)}} \left[\frac{\|Av\|_q}{\|Av\|_2} \right]^{\frac{1}{q-1}} \min_{\substack{x \in S  \\ x \ne 0}} \frac{\|x\|_2}{\|x\|_p} \bigg)\; \|A\|_{p \rightarrow q} $$ 
for all $p,q \ge 2$.
\end{theorem}

\begin{proof}
Using the bound $\|A\|_{p \rightarrow q} \le \|A\|_{1 \rightarrow q}^{\frac{2-p}{p}} \|A \|_{2 \rightarrow q}^{\frac{2(p-1)}{p}}$ and the results of Lemma \ref{lm:2toq}, our estimates for $\|A\|_{p \rightarrow q}$, $1 \le p \le 2 \le q$, follow relatively quickly:
\begin{align*}
\|A\|_{p \rightarrow q} &\le \|A\|_{1 \rightarrow q}^{\frac{2-p}{p}} \|A \|_{2 \rightarrow q}^{\frac{2(p-1)}{p}} \\
&\le \bigg[ \frac{\|A y\|_q}{\|y\|_1} \bigg]^{\frac{2-p}{p}} \bigg( \left[\frac{2\|Ar_j\|_\infty}{\|Ar_j\|_q} \right]^{\frac{q-2}{2(q-1)}} \left[\frac{\|Av\|_2}{\|Av\|_q} \right]^{\frac{1}{q-1}}\max_{\substack{x \in S \\ x \ne 0} } \, \frac{\|A x\|_q}{\|x\|_2}\bigg)^{\frac{2(p-1)}{p}} \\
&=\bigg[ \frac{\|y\|_p}{\|y\|_1}\, \frac{\|A y\|_q}{\|y\|_p} \bigg]^{\frac{2-p}{p}} \bigg( \left[\frac{2\|Ar_j\|_\infty}{\|Ar_j\|_q} \right]^{\frac{q-2}{2(q-1)}} \left[\frac{\|Av\|_2}{\|Av\|_q} \right]^{\frac{1}{q-1}}\max_{\substack{x \in S \\ x \ne 0} } \, \frac{\|x\|_p}{\|x\|_2} \, \frac{\|A x\|_q}{\|x\|_p}\bigg)^{\frac{2(p-1)}{p}} \\
&\le \bigg[ \frac{\|y\|_p}{\|y\|_1}\bigg]^{\frac{2-p}{p}} \bigg( \left[\frac{2\|Ar_j\|_\infty}{\|Ar_j\|_q} \right]^{\frac{q-2}{2(q-1)}} \left[\frac{\|Av\|_2}{\|Av\|_q} \right]^{\frac{1}{q-1}}\max_{\substack{x \in S \\ x \ne 0} } \, \frac{\|x\|_p}{\|x\|_2} \bigg)^{\frac{2(p-1)}{p}} \max_{\substack{x \in S \cup \{y\} \\ x \ne 0} } \,  \frac{\|A x\|_q}{\|x\|_p}.
\end{align*}
The regime of $p,q \ge 2$ is also straightforward, and follows from the bound $\|A\|_{p \rightarrow q} \le n^{1/2 - 1/p} \|A\|_{2 \rightarrow q}$:
\begin{align*}
\|A\|_{p \rightarrow q} &\le n^{\frac{p-2}{2p}} \|A\|_{2 \rightarrow q} \\ &\le n^{\frac{p-2}{2p}}\left[\frac{2\|Ar_j\|_\infty}{\|Ar_j\|_q} \right]^{\frac{q-2}{2(q-1)}} \left[\frac{\|Av\|_2}{\|Av\|_q} \right]^{\frac{1}{q-1}} \;\max_{\substack{x \in S \\ x \ne 0} } \, \frac{\|A x\|_q}{\|x\|_2} \\ &\le n^{\frac{p-2}{2p}} \left[\frac{2\|Ar_j\|_\infty}{\|Ar_j\|_q} \right]^{\frac{q-2}{2(q-1)}} \left[\frac{\|Av\|_2}{\|Av\|_q} \right]^{\frac{1}{q-1}}\max_{\substack{x \in S \\ x \ne 0} } \, \frac{\|x\|_p}{\|x\|_2} \; \max_{\substack{x \in S  \\ x \ne 0} } \,  \frac{\|A x\|_q}{\|x\|_p}.
\end{align*}
\end{proof}

Theorem \ref{thm:simple} follows immediately from the above theorem and the basic p-norm inequality $\|x\|_q \le \| x \|_p \le n^{1/p - 1/q} \|x\|_{q}$ for all $x \in \mathbb{C}^n$ and $p \le q$.

\begin{table}[t]
\centering
\begin{tabular}{|l || r | r | r | r| r| r | r| r|}
\hline  & & & \multicolumn{3}{c|}{Adj. Mat. $A$} & \multicolumn{3}{c|}{Norm. Adj. $\mathcal{A}$} \\
 graph & $n$ & $\text{deg}_{avg}$  & $\alpha_{apx}$ & $\alpha_{imp}$ & $M_{imp}$   & $\alpha_{apx}$ & $\alpha_{imp}$ & $M_{imp}$ \\
\hline \hline
144 & 144,649 & 14.86 & 2.11 & $0.5\%^\dagger$& $2.8\%$  & 3.30 & $22.1\%$& $0.0\%$\\ \hline
598a & 110,971 & 13.37 & 2.60 & $12.3\%$& $0.0\%$ & 3.20& $21.0\%$& $0.0\%$ \\ \hline
auto & 448,695 & 14.77 & 2.23 & $5.2\%$& $0.0\%$ & 3.40 & $25.1\%$& $5.1\%$ \\ \hline
ca-AstroPh & 18,772 & 21.10 & 1.91 & $0.0\%$& $0.0\%$ & 1.84 & $0.0\%$& $0.0\%$ \\ \hline
ca-HepPh & 12,008 & 19.74 & 1.43 & $0.0\%$& $0.0\%$ & 1.92 & $0.0\%$& $0.0\%$ \\ \hline
coPapersCiteseer & 434,102 & 73.88 & 1.33 & $13.1\%$& $0.0\%$ & 2.94 & $17.6\%$& $0.0\%$ \\ \hline
coPapersDBLP & 540,486 & 56.41  & 1.88 & $9.9\%$& $0.0\%$ & 3.51 & $24.5\%$& $0.0\%$ \\ \hline
data & 2,851 & 10.59 & 2.10 & $3.8\%$& $2.9\%$ & 2.33 & $8.9\%$& $3.2\%$ \\ \hline
fe\_rotor & 99,617 & 13.30 & 1.65 & $0.0\%$& $0.0\%$  & 3.13 & $21.0\%$& $2.2\%$\\ \hline
fe\_tooth & 78,136 & 11.58 & 2.58 & $11.6\%$& $0.0\%$ & 3.10 & $19.7\%$& $0.0\%$ \\ \hline
m14b & 214,765 & 15.64 & 2.43 & $9.4\%$& $0.0\%$ & 3.26 & $21.7\%$& $0.0\%$  \\ \hline
mycielskian16 & 49,151 & 679.18 & 2.38 & $11.3\%$& $6.8\%$ & 2.47 & $10.4\%$& $0.6\%$ \\ \hline
wing\_nodal & 10,937 & 13.80 & 2.42 & $9.1\%$& $0.0\%$ & 2.65 & $14.1\%$& $2.4\%$ \\ \hline

\end{tabular}
\vspace{4 mm}
\caption{We consider the matrix $2\to4$ norm for a variety of graphs taken from application. This table contains experiments for every undirected graph in the SuiteSparse Matrix Collection \cite{davis2011university} with between $1,000$ and $1,000,000$ vertices and average degree $\text{deg}_{avg}$ at least $10$. For either the adjacency matrix $A$ or the normalized adjacency matrix $\mathcal{A}$, we define a number of parameters. The quantity $\alpha_{apx}$ is the approximation produced by Lemma $\ref{lm:2toq}$, e.g., Lemma \ref{lm:2toq} produces a vector $x$ and a number $\alpha_{apx}$ such that $\|A x \|_4 / \|x \|_2 \ge \alpha_{apx} \|A\|_{2\to 4}$. The quantity $\alpha_{imp}$ is the percent decrease in the approximation between Lemma \ref{lm:2toq} and standard interpolation bounds applied to the vectors $v$ and $r$ achieving the $2\rightarrow 2$ and $2 \rightarrow \infty$ norms, respectively (e.g., the quantity $[\|Ar\|_\infty \|Av \|_2 / (\|Ar\|_4 \|Av\|_4)]^{1/2}$. The quantity $M_{imp}$ is the percent increase in largest computed $\|A x\|_4 / \|x\|_2$ between Lemma \ref{lm:2toq} and the test vectors $v$ and $r$ produced by interpolation. The additional computation of $\|A x \|_4 / \|x\|_2$ for the rows of $A$ led to an improvement in the best-known vector approximation over $30\%$ of the time. However, even when the computation did not produce an improved vector approximation, it provided an improved guarantee (e.g., approximation ratio) of the quality of the existing vector approximation over $80\%$ of the time.}
\label{tab:experiments}
\end{table}

\subsection{From Theory to Practice} To this point, our analysis has consisted entirely of the worst-case theoretical behavior. Theorem \ref{thm:complicated} gives a theoretical improvement over previous techniques. Here, we focus on how the theoretical results of this work influence the practical implementation of algorithms in practice. In terms of complexity, computing the largest singular value up to a relative accuracy of order $\epsilon$ with high probability requires at most $O(\log (n/\epsilon) /\epsilon^{1/2})$ matrix vector products using Krylov subspace methods (see \cite{musco2015randomized, urschel2021uniform,meyer2023unreasonable} and references therein for details). The ratio of $\|A r_j\|_q/\|r_j\|_p$ for all $j = 1,...,m$ can be computed in matrix multiplication time, and for a matrix $A \in \mathbb{C}^{m \times n}$ with at most $k$ non-zeros per row, in $O(k^2 n)$ time, making our results particularly attractive for sparse matrices. In addition, we note that the vector with largest $p \rightarrow q$ norm produced by Theorem \ref{thm:complicated} is a strong initial guess for the matrix $p \rightarrow q$ norm, but would certainly benefit from the application of a modified power method in the spirit of Boyd and Higham \cite{boyd1974power,higham1992estimating}, or a general non-linear solver. 

Here we consider the practical effectiveness of Lemma \ref{lm:2toq} for estimating the $2 \rightarrow 4$ norm of the adjacency matrix of sparse graphs. This is a natural choice, given the connection between the $2 \to 4$ norm and small set expansion in graph theory \cite{barak2012hypercontractivity}. Let $G = (V,E)$, $n = |V|$, be an undirected graph. We denote by $\text{deg}(i)$ the degree of vertex $i$, i.e., the number of edges containing vertex $i$. The adjacency matrix $A \in \mathbb{R}^{n \times n}$ of a graph $G$ has entries $A(i,j) = 1$ if there is an edge between vertex $i$ and vertex $j$, and $A(i,j) = 0$ otherwise. The normalized adjacency matrix $\mathcal{A}$ has entries $\mathcal{A}(i,j) = 1/[\text{deg}(i) \text{deg}(j)]^{1/2}$ if there is an edge between vertex $i$ and vertex $j$, and $\mathcal{A}(i,j) = 0$ otherwise. A graph is said to be d-regular if $\text{deg}(i) = d$ for all vertices $i$. In this setting, we note that $\mathcal{A} = A/d$. We compare the performance of our results (Lemma \ref{lm:2toq}) to interpolation in two different settings.

First, we consider $d$-regular graphs, a setting for which Lemma \ref{lm:2toq} will clearly outperform interpolation techniques. Because each row contains exactly $d$ entries equal to one and all other entries zero, the row norm contains no information save for the degree $d$. In addition, by the Perron-Frobenius Theorem, the top singular value is always $d$ and the constant vector is always a corresponding right singular vector. Therefore, interpolation techniques cannot reliably discern any structure from a d-regular graph, save for the fact that it is d-regular. In contrast, using the rows as test vectors gives further information.
 
Next, we consider large real-world graphs with non-uniform degree structure, taken from the SuiteSparse matrix collection (formerly referred to as the University of Florida Sparse Matrix Collection) \cite{davis2011university}. This collection contains graphs from a variety of applications, including structural finite element meshes, social networks, collaboration networks, road networks, and many other real-world examples. We consider every undirected graph from this collection with between $1,000$ and $1,000,000$ vertices and average degree at least $10$. We analyze both the adjacency and normalized adjacency matrices of these graphs, and report the performance of Lemma \ref{lm:2toq}. In a number of cases, a row without maximum 2-norm provides the best approximation. More generally, experimentally we note that even if the testing of the rows of a matrix do not improve the estimate of $\|A\|_{2 \to 4}$ itself, it often still improves the approximation guarantee -- an equally important parameter that verifies the quality of the estimate. We provide the details of these experiments in Table \ref{tab:experiments}.

{ 
	\bibliographystyle{plain}
	\bibliography{main} }

\end{document}